%% file: paper.tex
\documentclass[a4paper,11pt]{article}
\usepackage[dvips]{color}
\usepackage[section]{algorithm}
\usepackage{algorithmic}
\usepackage{paper,latexsym,multicol,macros}
\usepackage{trbonn}
\usepackage{changebar}
\usepackage{picinpar}
\usepackage{wasysym}

\def\TRIA{{\mbox{$\mbox{}_{\triangle}$}}}
\def\HEXA{{\mbox{$\mbox{}_{\mbox{\scriptsize \hexagon}}$}}}
\def\QUAD{{\mbox{$\mbox{}_{\Box}$}}}

\def\assumptionname{Conjecture}

\title{Exploring Simple Triangular and Hexagonal Grid Polygons Online}

\author{Daniel Herrmann\thanks{%
        University of Bonn,
        Institute of Computer Science, Dept.~I,
        R\"omerstr.~164,
        53117~Bonn, Germany.
}
\and    Tom Kamphans\thanks{%
Braunschweig University of Technology,
Computer Science, Algorithms Group,
M{\"u}hlenpfordtstra{\ss}e 23,
38106 Braunschweig, Germany
}
\and    Elmar Langetepe\footnotemark[1]
}

\date{December 2007\\[5pt]}
\reportnumber{007}

\begin{document}
\maketitle

{\pagestyle{empty}
\mbox{}\newpage

\begin{abstract}
We investigate the online exploration problem (aka covering)
of a short-sighted mobile robot
moving in an unknown cellular environment with
hexagons and triangles as types of cells. To explore a
cell, the robot must enter it. Once inside, the robot knows which
of the 3 or 6 adjacent cells exist and which are boundary edges.
The robot's task is to visit every cell in the given environment
and to return to the start. Our interest
is in a short exploration tour; that is, in keeping the number
of multiple cell visits small.
For arbitrary environments containing no obstacles
we provide a strategy producing tours of length
$S \leq C + \frac14 E - 2.5$
for hexagonal grids, and
$S \leq C + E - 4$
for triangular grids.
$C$ denotes the number of cells---the area---, $E$ denotes
the number of boundary edges---the perimeter---of the given environment.
Further, we show that our strategy is $\frac43$-competitive in both
types of grids, and we provide lower bounds of $\frac{14}{13}$ for
hexagonal grids and $\frac76$ for triangular grids.

\smallskip
The strategies were implemented in a Java applet \cite{hiklm-gaesu-00}
that can be found in

\centerline{\tt http://www.geometrylab.de/Gridrobot/}

\medskip

\end{abstract}

\mbox{}\newpage\mbox{}
\clearpage
}
\pagenumbering{arabic}

\input{introduction.tex}

\input{definitions.tex}

\input{lowerbound.tex}

\input{SmartDFSStrat.tex}

\input{SmartDFSAnalysisHex.tex}

\clearpage
\input{SmartDFSAnalysisTria.tex}

\input{summary.tex}

\bibliographystyle{abbrv}
\bibliography{%
        ../../../abt1/biblio/local,%
        ../../../abt1/biblio/update,%
        ../../../abt1/biblio/geom}

\end{document}

%% file: introduction.tex
\sect{Introduction}{intro}
Exploring an unknown environment 
is one of the basic tasks of autonomous
mobile robots and has received a lot of attention
in computational geometry and in robotics; 
see, for example, \cite{dkp-hlue1-98,hikk-pep-01,m-spn-97,m-gspno-00,rksi-rnut-93,clhkbk-prmta-05,b-olsn-98,bs-eewl-03,dmr-mrcre-01}---just to
mention a few of these works.

For some applications, it is convenient to subdivide the given environment
by a regular grid into basic blocks (so-called {\em cells}). For example,
the agent's vision may be limited and a cell is used as to approach the
visibility range. Or the agent has to visit every part of the environment
for cleaning or lawn mowing, and a cell is an approximation of the 
robot's tool (sometimes, this task is called {\em covering}).
The robot's position is always given by the cell currently occupied by the
robot. From its current position, the robot
can enter one of the neighboring {\em free} cells (i.e., cells that are
not blocked by an obstacle).
The whole environment is not known in advance---so we are dealing with
{\em online} strategies---, but once inside a cell, the robot knows
which neighboring cell is blocked and which one is free.
The robot's task is to visit every free cell
inside the given environment and to return to the start.
There are only three possible regular tilings of the plane:
square, hexagonal, or triangular subdivisions \cite{bc-mre-87}. 
We call a subdivision of the
given environment into squares (hexagons, triangles) a square polygon 
(hexagonal polygon, triangular polygon; respectively).
Hexagonal cells are a matter of particular interest
for robots that are equipped with a circular tool such as lawn mowers, 
because hexagonal grids provide a better approximation for the tool than
square grids \cite{afm-aalmm-00}.

In a square polygon with obstacles, the offline problem (i.e.,
finding a minimum length tour that visits every cell)
is known to be NP-hard, by work of Itai et
al.~\cite{ips-hpgg-82}. By modeling the environment as a {\em grid graph}
with one vertex for every cell and edges between neighboring cells,
we can use $1+\epsilon$ approximation schemes for Euclidean TSP
by Grigni et al.~\cite{gkp-aspgt-95}, Arora~\cite{a-ptase-96},
and Mitchell~\cite{m-gsaps-96}. For square polygons there is a
$\frac{53}{40}$ approximation by Arkin et al.~\cite{afm-aalmm-00}.

In a square polygon without obstacles, the complexity of constructing
off\-line a minimum length tour is still open.
Ntafos~\cite{n-wrlv-92} and Arkin et al.~\cite{afm-aalmm-00}
have shown how to approximate the minimum length tour with
factors of $\frac{4}{3}$ and $\frac{6}{5}$, respectively.
Umans and Lenhart~\cite{ul-hcsgg-97} have provided an $O(C^4)$ 
algorithm for deciding if there exists a Hamiltonian cycle 
(i.e., a tour that visits each of the $C$ cells
of a polygon {\em exactly} once).
For the related problem of Hamiltonian paths (i.e., different
start and end positions), Everett~\cite{e-hpnrg-86} has given a polynomial algorithm
for certain grid graphs. Cho and Zelikovsky~\cite{hz-scthc-95}
studied {\em spanning closed trails}.
Hamiltonian cycles on triangular and hexagonal grids were studied by
Polishuk et al.~\cite{pam-hctg-06,amp-tnchg-07},
and Islam et al.~\cite{imrrx-hchgg-07}, see also
\cite{afimmmprx-nbsts-07}.

In this paper, our interest is in the online version of the cell exploration
problem for {\em hexagonal} and {\em triangular} polygons without holes. 
The task of exploring square polygons with 
holes was independently considered by Gabriely and Rimon
\cite{gr-colcg-03} and
Icking et al.~\cite{ikkl-esgp-05}, see also Kamphans~\cite{k-maole-05}.
Our exploration strategy is based on the strategy {\em SmartDFS} by
Icking et al.~\cite{ikkl-egpol-05} for simple polygons. This strategy is
$\frac43$-competitive%
\footnote{That is, the path produced by this online strategy is never longer 
than $\frac43$ times the optimal offline path.}
and the number of
steps from cell to cell is bounded by $C+ \frac12E -3$, where $C$ denotes
the number of cells (i.e., the polygon's area) and $E$ the number of edges
(the polygon's perimeter). Further, there is a lower bound of $\frac76$ 
on the competitive factor for this problem.

Another online task is the {\em piecemeal exploration}, 
where the robot has to interrupt the exploration every 
now and then so as to return to the
start point, for example, to refuel.
Piecemeal exploration of grid graphs was studied by
Betke et~al.~\cite{brs-plue-95} and Albers et~al.~\cite{aks-eueo-02}.
Note that their objective is to visit every node {\em and} every edge,
whereas we require a complete coverage of only the cells.
Subdividing the robot's environment into grid cells is used also
in the robotics community, see, for example,
Moravec and Elfes~\cite{me-hrmwa-85},
Elfes~\cite{e-uogmr-89},
Bruckstein et al.~\cite{blw-dcaru-99,wlb-mvpdr-00},
and Koenig and Liu \cite{kl-tcars-01}.
See also the survey by Choset \cite{c-crsrr-01}.

Our paper is organized as follows:
In \refsect{defi}, we give more detailed description of our explorer and the
environment. We give lower bounds on the competitive factor
in \refsect{lowerbounds}.
In \refsect{smartdfs}, we present an exploration strategy for simple 
polygons. We analyze the performance of this strategy in hexagonal polygons
in \refsect{hexanalysis} 
and for triangular polygons in \refsect{trianalysis}.

%% file: definitions.tex
\sect{Definitions}{defi}
\pstexfig{(i) Polygon with 23 cells, 38 edges and one(!) hole (black cells), 
a path from $s$ to $t$ of length 6
(ii)-(iv) neighboring and touching cells; the agent 
can determine which of the neighboring cells (marked by an arrow) are
free, and enter an adjacent free cell.}
{figs/CellPoly}

\begin{defi}{gridpoly}
We consider polygons that are subdivided by a regular grid. 
A \DEFI{cell}~ is a basic block in our environment.
A cell is either \DEFI{free} and can be visited by the
robot, or \DEFI{blocked}  (\IE, unaccessible for the robot).%
\footnote{In the following, we sometimes use the terms 
{\em free cells} and {\em cells} synonymously.}
We call two cells 
\DEFI{adjacent} or \DEFI{neighboring} if they share a common edge,
and \DEFI{touching} if they share only a common corner.

A \DEFI{path}, $\pfad$, from a cell $s$ to a cell $t$ is 
a sequence of free cells $s=c_1, \ldots , c_n=t$ 
where $c_i$ and $c_{i+1}$
are adjacent for $i=1,\ldots,n-1$.
Let $\len{\pfad}$ denote the length of $\pfad$.
We assume that the cells have unit size, so the length of the path
is equal to the number of \DEFI{steps} from cell to cell that the robot walks.

A \DEFI{grid polygon}, $P$, is a path-connected set of free cells;
that is, for every $c_1, c_2\in P$ exists a path from $c_1$ to $c_2$
that lies completely in $P$. We denote a grid polygon subdivided
into square, hexagonal, or triangular cells by
$P\!\QUAD$, $P\!\HEXA$, or $P\!\TRIA$, respectively.

We call a set of touching blocked cells that are
completely surrounded by free cells an \DEFI{obstacle}
or \DEFI{hole}; see \reffig{figs/CellPoly}.
Polygons without holes are called \DEFI{simple polygon}\DEFI{s}.
\end{defi}

\pstexfig{The perimeter, $E$, is used to distinguish between {\em thin} and
{\em thick} environments.}{figs/fleshyskinny}

We analyze the performance of an exploration strategy using
some parameters of the grid polygon.
In addition to the area, $C$, of a polygon we use
the {\em perimeter}, $E$. The parameter $C$ is the number of free cells and
$E$ is the total number of edges that appear between a free cell and a blocked
cell; see, for example, \reffig{figs/CellPoly} 
or \reffig{figs/fleshyskinny}.
We use the perimeter, $E$, to distinguish between thin environments
that have many corridors of width 1, and thick environments that
have wider areas.
In the following sections we present strategies
that explore grid polygons using no more than
roughly $C+\frac14 E$ steps (hexagons) and $C+E$ (triangles).
Since all cells in the environment have to be visited, $C$ is 
a lower bound on the
number of steps that are needed to explore the whole polygon and
to return to $s$.
Thus,  the number of edges (or a fraction of them)
is an upper bound for the number of additional cell visits.
For thick environments, the value of $E$ is in $O(\sqrt{C})$, so that 
the number of additional cell visits is
substantially smaller than the number of free cells. Only for polygons
that do not contain any $2\times2$ square of free cells,
$E$ achieves its maximum value of $2(C+1)$. But in thoses cases, no
online strategy can do better; even the optimal path has this length.

\medskip

We will see that our strategy SmartDFS
explores the polygon in layers: Beginning with the cells along the
boundary, the agent proceeds towards the interior of $P$. 
Thus, we number the single layers:

\begin{defi}{layers} Let $P$ be a (simple) grid polygon
(of either type). The 
boundary cells of $P$ uniquely define the {\em first layer} of $P$. 
The polygon 
$P$ without its first layer is called the {\em 1-offset} of $P$. 
The \ith{$\ell$} layer and the $\ell$-offset of $P$ are defined 
successively; see \reffig{figs/doffset}.
\end{defi}

\pstexfig{The 2-offset (shaded) of a grid polygon $P$.}{figs/doffset}

Note that the $\ell$-offset of a polygon $P$ is not necessarily
connected.

%% file: lowerbound.tex
\sect{Lower Bounds}{lowerbounds}
In an online setting, the agent does not know the environment in advance.
So we are intested in the best competitive factor we can expect for a
strategy that visits every cell at least once and returns to the 
start cell.
In an environment {\em with} holes, we have the following theorem:

\begin{theo}{ExplCompl}
The competitive complexity of exploring an unknown grid polygon with
obstacles and hexagonal or triangular cells is 2.
\end{theo}
\begin{proof}
We can simply adapt the lower bound construction for square cells by
Icking et al.~\cite{ikkl-esgp-05}, yielding a lower bound of 2.
On the other hand, we can apply a simple depth-first search, resulting
in a tour with $2C-2$ steps.
The shortest tour needs at least $C$ steps to visit all cells
and to return to $s$, so DFS is competitive with a factor of 2.
\end{proof}

Surprisingly, we cannot trim the 
lower bound construction by Icking et al.~\cite{ikkl-esgp-05}
for simple polygons with hexagonal or triangular cells. 
The lower bound construction for polygons with holes uses 
only corridors of width 1, so the
type of cells does not matter. In contrast, the construction 
for simple polygons uses wider areas, where the number of neighboring cells
plays a major role. However, the lower bounds for squares and triangles
are identical:

\pstexfig{A lower bound on the exploration of simple triangular polygons.
The thin dashed lines show the optimal solution, the bold dashed triangles
denotes the start cell of the next block.}
{figs/triLowerBound}

\begin{theo}{triLB}
There is no online strategy
for the exploration of simple triangular 
grid polygons with a competitive factor better than
${7\over 6}$.
\end{theo}
\begin{proof}
Let the agent start in a cell with two neighbors, one to the south and one to 
the northwest, see \reffig{figs/triLowerBound}(i). If it walks to the south,
we add a cell such that the only possible step is to the southwest;
see \reffig{figs/triLowerBound}(ii).
If it walks from the start to the east,
 we force it to move another step to the east; see
\reffig{figs/triLowerBound}(iii).
In both cases, the agent has the choice to leave the polygon's boundary
(\reffig{figs/triLowerBound}(iv) and (vi)) or to follow the polygon's 
boundary (\reffig{figs/triLowerBound}(v) and (vii)). 
In either case, we fix the polygon after this step.
If the agent leaves the
boundary, it needs at least 12 steps while the optimal path has
10 steps. In the other case, the agent needs at least 26 steps; the optimal
path in this polygon has a length of 22 steps.

To construct arbitrarily large polygons, we use more of these blocks 
and glue them together using the cell that is shown with bold dashed lines 
in the figure and
denotes the start cell of the next block. Unfortunately, both the online
strategy and the optimal path need two additional steps for the 
transition between two blocks. Let $n$ denote the number of 
blocks, then we have in the best case a ratio of
${26 + 28(n-1) \over 22+24(n-1)}$, which converges to $\frac76$ if $n$ goes
to infinity.
\end{proof}

\bigskip
The lower bound construction for hexagonal polygons is simpler, but
yields a smaller value.

\pstexfig{A lower bound on the exploration of simple polygons.
The dashed lines show the optimal solution.}
{figs/hexLowerBound}

\begin{theo}{hexLB}
There is no online strategy
for the exploration of simple hexagonal 
grid polygons with a competitive factor better than
${14\over 13}$.
\end{theo}
\begin{proof}
We start in a cell with four neighboring cells, see 
\reffig{figs/hexLowerBound}(i). The agent may leave the polygon's boundary
by walking northwest or southwest, or follow the boundary by walking
north or south. In the first case, we close the block as shown in
\reffig{figs/hexLowerBound}(ii), in the second case as shown in
\reffig{figs/hexLowerBound}(iii), yielding a ratio of 
$\frac76$ or $\frac{13}{12}$, respectively.

As in the preceeding proof, we construct polygons of arbitrary size by
concatenating the blocks from \reffig{figs/hexLowerBound}(ii)
and \reffig{figs/hexLowerBound}(iii).
A subsequent block attaches using the cell(s) shown with bold, dashed lines.
Again, we need one or two additional steps for the transition, yielding
a best-case ratio of
$13+14(n-1) \over 12+13(n-1)$, where $n$ denotes the number of blocks.
This ratio converges to $\frac{14}{13}\approx 1.076$.
\end{proof}

%% file: SmartDFSStrat.tex
\sect{Exploring Simple Polygons}{smartdfs}
In this section, we briefly describe the strategy 
{\em SmartDFS}\TRIA,\HEXA.
Our strategy is based on the same ideas as
SmartDFS\QUAD \cite{ikkl-esgp-05}, but it is generalized for
trianguar and hexagonal grids.

The basic idea is to use a simple DFS strategy as shown in \refalgo{dfs}.%
\footnote{The command \Befehl{move{\em (dir)}} executes the actual motion of
the agent.
The function \Befehl{un\-ex\-plored{\em (dir)}} returns
true, if the cell in the given direction seen from the agent's 
current position is not yet visited, and false otherwise.
Given a direction {\em dir}, \Befehl{reverse{\rm (dir)}} returns the direction
turned by $180^\circ$.}
From the current position, the explorer
tries to visit the adjacent cells in clockwise order, see the procedure 
{\em ExploreCell}.
If the adjacent cell is still unexplored, the agent enters this cell,
proceeds recursively with the exploration, and walks back, see the procedure
{\em ExploreStep}. 
Altogether, the polygon is explored following the {\em left-hand rule}:
The agent always keeps the
polygon's boundary or the explored cells on its left side.

\begin{algorithm}
\AlgoCaption{DFS}{dfs}
\begin{description}
\item[DFS(\VAR{P}, \VAR{start}):]~
\begin{algorithmic}
\STATE Choose direction \VAR{dir} such that the cell behind 
\STATE\mbox{}\quad the explorer is blocked;
\STATE ExploreCell(\VAR{dir});
\end{algorithmic}

\item[ExploreCell(\VAR{dir}):]~
\begin{algorithmic}
  \STATE \COMMENT{Left-Hand Rule:}
  \FORALL{ cells $c$ adjacent to the current cell, in clockwise order starting with the cell opposite to \VAR{dir} } 
  \STATE $\VAR{newdir} :=$ Direction towards $c$;
  \STATE ExploreStep(\VAR{newdir});
  \ENDFOR
\end{algorithmic}

\item[ExploreStep(\VAR{dir}):]~
\begin{algorithmic}
\IF{unexplored(\VAR{dir})}
  \STATE move(\VAR{dir});
  \STATE ExploreCell(\VAR{dir});
  \STATE move(reverse(\VAR{dir}));
\ENDIF
\end{algorithmic}
\end{description}
\end{algorithm}

Obviously, all cells are visited, because the polygon is connected; and
the whole path consists of $2C-2$ steps, because each cell---except
for the start---is entered exactly once by the first \Befehl{move} statement,
and left exactly once by the second \Befehl{move} statement
in the procedure {\em ExploreStep}.

\pstexfig{First improvement to DFS: Return directly to those cells that
still have unexplored neighbors.}{figs/dfsverbesserung1}

The first improvement to the simple DFS is to return directly to those
cells that have unexplored neighbors. 
See, for example, \reffig{figs/dfsverbesserung1}:
After the agent has reached the cell $c_1$, DFS walks to $c_2$ through the
completely explored corridor of width 2. A more efficient return path
walks on a shortest path from $c_1$ to $c_2$. Note that the agent
can use for this shortest path only cells that are already known.
With this modification, the agent's position
might change between two calls of {\em ExploreStep}. Therefore, the procedure 
{\em ExploreCell} has to store the current position, and the agent has to
walk on the shortest path to this cell, see the procedure {\em ExploreStep}
in \refalgo{smartdfs}.%
\footnote{
The function \Befehl{unexplored{\em (cell}, {\em dir)}} returns true, if the
cell in direction {\em dir} from {\em cell} is not yet visited.}

\pstexfig{Second improvement to DFS: Detect polygon splits.}
{figs/dfsverbesserung2}

Now, observe the polygon shown in \reffig{figs/dfsverbesserung2}. 
DFS completely surrounds the polygon, returns to $c_2$ and explores
the left part of the polygon. Then, it walks to $c_1$ and
explores the right part. Altogether, the agent walks four times
through the narrow corridor. A more clever solution would
explore the right part immediately after the first visit of $c_1$, and
continue with the left part afterwards. This solution would walk
only two times through the corridor in the middle! 
The cell $c_1$ has the property that the graph of unvisited cells
splits into two components after $c_1$ is explored. We call cells
like this \DEFI{split cell}\DEFI{s}. 
The second improvement to DFS is to recognize split cells and 
diverge from the left-hand rule when a split cell is detected.
Essentially, we want to split the set of cells into several
components, which are finished in the reversed order
of their distances to the start cell.
The detection and handling of split cells is specified in 
\refsect{hexanalysis} and \refsect{trianalysis}.
\refalgo{smartdfs} resumes both improvements to DFS.

\pstexfig{Straightforward strategies are not better than SmartDFS.}
{figs/SimpleStrats}

Note that the straightforward strategy {\em Visit all boundary cells and
calculate the optimal offline path for the rest of the polygon}
does not achieve a competitive factor better than $2$. For example, in 
\reffig{figs/SimpleStrats}(i) this strategy visits almost every 
boundary cell twice, whereas SmartDFS visits only one cell twice.
Even if we extend the simple strategy to detect split cells while visiting
the boundary cells, we can not
achieve a factor better than $\frac43$. A lower bound on the
performace of this strategy is a corridor of width 3, see
\reffig{figs/SimpleStrats}(ii).
Moreover, it is not known whether the offline solution is NP-hard for
simple polygons.

\begin{algorithm}
\AlgoCaption{SmartDFS}{smartdfs}
\begin{description}
\item[SmartDFS(\VAR{P}, \VAR{start}):]~
\begin{algorithmic}
\STATE Choose direction \VAR{dir} such that the cell behind 
\STATE\mbox{}\quad the explorer is blocked;
\STATE ExploreCell(\VAR{dir});
\STATE Walk on the shortest path to the start cell;
\end{algorithmic}

\item[ExploreCell(\VAR{dir}):]~
\begin{algorithmic}
\STATE Mark the current cell with the number of the current layer;
\STATE \VAR{base} $:=$ {\rm current position};
\IF{not isSplitCell(\VAR{base})}
  \STATE \COMMENT{Left-Hand Rule:}
  \FORALL{ cells $c$ adjacent to the current cell, in clockwise order starting with the cell opposite to \VAR{dir} } 
  \STATE $\VAR{newdir} :=$ Direction towards $c$;
  \STATE ExploreStep(\VAR{base}, \VAR{newdir});
  \ENDFOR
\ELSE
  \STATE \COMMENT{choose different order, see \refpage{comporder}}\,ff
  \STATE Determine the types of the components using the layer numbers
  \STATE ~~~of the surrounding cells;
  \IF { No component of type III exists }
     \STATE Use the left-hand rule, but omit the first possible step.
  \ELSE
     \STATE Visit the component of type III at last.
  \ENDIF
\ENDIF
\end{algorithmic}

\item[ExploreStep(\VAR{base}, \VAR{dir}):]~
\begin{algorithmic}
\IF{unexplored(\VAR{base}, \VAR{dir})}
  \STATE Walk on shortest path using known cells to $base$;
  \STATE move(\VAR{dir});
  \STATE ExploreCell(\VAR{dir});
\ENDIF
\end{algorithmic}
\end{description}
\end{algorithm}

%% file: SmartDFSAnalysisHex.tex
\ssect{The Analysis of SmartDFS\HEXA}{hexanalysis}
In this section, we analyze the performance of our strategy in
a hexagonal grid.
We start with an important property of the $\ell$-offset:

\begin{lem}{hexoffset}
The $\ell$-offset of a simple, hexagonal grid 
polygon, $P\!\HEXA$, has at least $12\ell$ edges fewer than $P\!\HEXA$.%
\footnote{Provided that the $\ell$-offset is not empty anyway.}
\end{lem} 

\pstexfig{The 2-offset (shaded) of a grid polygon $P\!\HEXA$.}
{figs/hexdoffset}

\begin{proof}
First, we cut the parts of the polygon
$P\!\HEXA$ that do not affect the $\ell$-offset.
Now, we observe the closed path, $\pfad$, that connects the midpoints
of the boundary cells of the remaining polygon in clockwise order; see
\reffig{figs/hexdoffset}. For every $60^\circ$ left turn the offset gains 
at most $2\ell$ edges and for every $60^\circ$ right turn the offset 
looses at least $2\ell$ edges (cutting off the passages that are 
narrower than $2\ell$ ensures that we have enough edges to loose at 
this point). 

\pstexfig{(i) A convex hexagonal grid polygon with six $60^\circ$
right turns in the first layer, (ii) adding two $60^\circ$ left turns 
forces adding to $60^\circ$ right turns.}
{figs/hexConvReflex}

It is easy to see that there are six more $60^\circ$ 
right turns than left turns (we count $120^\circ$ turn as two $60^\circ$
turns): In a convex polygon, the tour along the first layer has
six $60^\circ$ right turns. 
For every $60^\circ$ left turn that we add to the polygon, we
also add another $60^\circ$ right turns, see \reffig{figs/hexConvReflex}.
Altogether, we loose at least 12 edges for every layer.
\end{proof}

\pstexfig{A decomposition of $P\!\HEXA$ at the split cell $c$ and 
its handling in SmartDFS\HEXA.}
{figs/hexDecomposition1}


%
%
Now, let us consider the handling of a split cell. Observe a
situation as shown in \reffig{figs/hexDecomposition1}(i):
SmartDFS\HEXA\ has just met the
first split cell, $c$, in the second layer of $P\!\HEXA$. 
$P\!\HEXA$ divides into three parts:
$$P\!\HEXA = K_1 \disjoint K_2 \disjoint\ \{\,\mbox{visited cells of } P\!\HEXA\,\},$$
where $K_1$ and $K_2$ denote the connected components of the set of unvisited
cells.
In this case it is reasonable to explore the component $K_2$ first, because
the start cell $s$ is closer to $K_1$; that is, we can extend
$K_1$ with $\ell$ layers, such that the resulting polygon contains 
the start cell $s$.

%
%
More generally, we want to divide our polygon $P\!\HEXA$ into two parts, $P_1$
and $P_2$, such that each of them is an extension of the two
components. Both polygons overlap in the area around the split cell $c$.
At least one of these polygons contains the start cell. If only one
of the polygons contains $s$, we want our strategy to 
explore this part at last, expecting that in this part the path from the 
last visited cell
back to $s$ is the shorter than in the other part.
Vice versa, if there is a polygon that does {\em not} contain $s$, we explore 
the corresponding component first. In \reffig{figs/hexDecomposition1},
SmartDFS\HEXA\ recursively enters $K_2$, returns to 
the split cell $c$, and explores the component $K_1$ next.

%
%
In the preceding example, there is only one split cell in $P\!\HEXA$, but
in general there will may be a sequence of split cells, $c_1,\ldots,c_k$.
In this case, we apply the handling of split cells in a recursive way;
that is, if a split cell $c_{i+1}, 1\leq i< k$, is detected in one of the two 
components occurring at $c_i$ we proceed the same way as described earlier. 
If another split cell occurs in $K_2$, 
the role of the start cell is played by the preceding 
split cell $c_i$. In the following, the term {\em start cell}\/ always
refers to the start cell of the current component; that is, either
to $s$ or to the previously detected split cell. 
Note that---in contrast to square grid polygons---there is no case where
three components arise at a split cell apart from the start cell $s$.

%
%

\pstexfig{Several types of components.}{figs/hexSplitcell}

\paragraph*{Visiting Order}~\\
We use the layer numbers to decide which component we have 
to visit at last. Whenever a split cell occurs in layer $\ell$,
every component is one of the following\labelpage{comporder}
types, see \reffig{figs/hexSplitcell}:

\begin{romanlist}
\item[I.] $K_i$ is {\em completely} surrounded by layer $\ell$%
\footnote{More precisely, the part of layer $\ell$ that surrounds $K_i$ is
completely visited. For convenience, we use the slightly sloppy, but shorter 
form.}
\labelpage{comptypes}
\item[II.] $K_i$ is {\em not} surrounded by layer $\ell$
\item[III.] $K_i$ is {\em partially} surrounded by layer $\ell$
\end{romanlist}


\medskip
It is reasonable to explore the component of type III at last:
There are two cases in which SmartDFS\HEXA\ switches from a layer $\ell-1$ 
to layer $\ell$. Either it reaches the first cell of
layer $\ell-1$ in the current component and thus passes the start 
cell of the current component%
, or it hits another cell of 
layer $\ell-1$ but no polygon split occurs.
In the second case, the considered start cell must be located in a 
corridor that is completely explored; otherwise, the strategy would
be able to reach the first cell of layer $\ell-1$ as in the first case.
In both cases the part of $P\!\HEXA$ surrounding a component of type III 
contains the first cell of the current layer $\ell$
as well as the start cell.
Thus, we explore a component of type~III at first, provided that such a 
component exists.

\bigskip
Unfortunately, there are two cases in which no component of type III exists:
\begin{enumerate}
\item The part of the polygon that contains the preceding start cell is
  explored completely, see for example \reffig{figs/hexSplitcell2}(i).
  In this case the order of the components makes no difference.
%

\item Both components are completely surrounded by a layer,
  because the polygon split and the switch from one layer to the
  next occurs within the same cell, see \reffig{figs/hexSplitcell2}(ii).
  A step that follows the left-hand rule will move towards the
  start cell, so we just omit the first possible. 
\end{enumerate}

We proceed
with the rule in case~2 whenever there is no component of type III,
because the order in case 1 does not make a difference.

\pstexfig{No component of type III exists.}{figs/hexSplitcell2}

%
%
\clearpage 
\sssect{An Upper Bound on the Number of Steps}{ubhex}
For the analysis of our strategy we consider two 
polygons, $P_1$ and $P_2$, as follows.
Let $Q$ be the polygon that is made of $c$ and extended by $q$ layers,
where
$$q:=\cases{\ell, & if $K_2$ is of type I\cr \ell-1, 
                  & if $K_2$ is of type II}\,.$$
$K_2$ denotes the component that is explored first, and
$\ell$ denotes the layer in which the split cell was found.
We choose $P_2 \subset P\!\HEXA\cup Q$ such that $K_2 \cup \{c\}$ 
is the $q$-offset of $P_2$, 
and $P_1 := ((P\!\HEXA\backslash P_2) \cup Q) \cap P\!\HEXA$,
see \reffig{figs/hexDecomposition1}.
The intersection with $P\!\HEXA$ is necessary, because $Q$ may exceed the boundary
of $P\!\HEXA$, see \reffig{figs/hexDecomposition2}. 
Note that at least $P_1$ contains the preceding start cell.
There is an arbitrary number of polygons $P_2$, such that
$K_2 \cup \{c\}$ is the $q$-offset of $P_2$, because 'dead ends' of $P_2$ 
that are not wider than $2q$ do not affect the $q$-offset.
To ensure a unique choice of $P_1$ and $P_2$, we require
that both $P_1$ and $P_2$ are connected, and both 
$P\!\HEXA \cup Q=P_1 \cup P_2$ 
and $P_1 \cap P_2 \subseteq Q$ are satisfied.

The choice of $P_1, P_2$ and $Q$ ensures that the agent's path
in $P_1\backslash Q$ and in $P_2\backslash Q$ do not change 
compared to the path in $P\!\HEXA$. The parts of the agent's path
that lead from $P_1$ to $P_2$ and from $P_2$ to $P_1$
are fully contained in $Q$.
Just the parts inside $Q$ are bended to connect the appropriate
paths inside $P_1$ and $P_2$;
see \reffig{figs/hexDecomposition1} and \reffig{figs/hexDecomposition2}.

\pstexfig{The component $K_2$ is of type I. $Q$ may exceed
$P\!\HEXA$.}{figs/hexDecomposition2}

In \reffig{figs/hexDecomposition1}, $K_1$ is of type III and $K_2$ is 
of type II. A component of type I occurs, if we detect a split cell
as shown in \reffig{figs/hexDecomposition2}.
Note that $Q$ may exceed $P\!\HEXA$, but $P_1$ and $P_2$
are still well-defined.

\bigskip

%
%
We want to visit every cell in the polygon and to return to
$s$. Every  strategy needs at least $C(P\!\HEXA)$ steps to fulfill this task,
where $C(P\!\HEXA)$ denotes the number of cells in $P\!\HEXA$.
Thus, we can split the overall length of the exploration path, $\pfad$,
into two parts, $C(P\!\HEXA)$ and $\excess(P\!\HEXA)$, 
with $|\pfad|=C(P\!\HEXA)+\excess(P\!\HEXA)$.
In this context,
$C(P\!\HEXA)$ is a lower bound on the number of steps that are needed for
the exploration task, whereas
$\excess(P\!\HEXA)$ is the number of additional cell visits.

Because SmartDFS\HEXA\ recursively explores $K_2\cup\{c\}$, we want to
apply the upper bound inductively to the component $K_2\cup \{c\}$. 
If we explore $P_1$ with SmartDFS\HEXA\ until $c$ is met, 
the set of unvisited cells in $P_1$ is equal to $K_1$, because the 
path outside $Q$ do not change. Thus, we can apply
our bound inductively to $P_1$, too. 
The following lemma gives us the relation between
the path lengths in $P\!\HEXA$ and the path lengths in the two components.

\begin{lem}{hexcomponent}
Let $P\!\HEXA$ be a simple hexagonal grid polygon. 
Let the explorer visit the first 
split cell, $c$, which splits the unvisited cells of $P\!\HEXA$ 
into two components 
$K_1$ and $K_2$, where $K_2$ is of type I or II. 
With the preceding notations we have
$$\excess(P\!\HEXA)\leq \excess(P_1)+ \excess(K_2\cup\{c\})+1\; .$$
\end{lem}
\begin{proof}
The strategy SmartDFS\HEXA\ has reached the split cell $c$ and 
explores $K_2\cup\{c\}$ with start cell $c$ first. Because $c$ is the
first split cell, there is no excess in $P_2\backslash (K_2\cup\{c\})$
and it suffices to consider $\excess(K_2\cup\{c\})$ for this
part of the polygon. 
After $K_2\cup\{c\}$ is finished, the agent returns to $c$ 
and explores $K_1$. For this part we take 
$\excess(P_1)$ into account. 
Finally, we add one single step, because the split cell $c$
is visited twice: once, when SmartDFS\HEXA\ detects the split and
once more after the exploration of $K_2\cup\{c\}$ is
finished. Altogether, the given bound is achieved. 
\end{proof}

\medskip

$c$ is the first split cell in $P\!\HEXA$, so $K_2\cup\{c\}$ is the $q$-offset of
$P_2$ and we can apply \reflem{hexoffset} to bound
the number of boundary edges of $K_2\cup\{c\}$ by the number
of boundary edges of $P_2$. The following lemma
allows us to charge the number of edges in $P_1$ and $P_2$
against the number of edges in $P\!\HEXA$ and $Q$.

\begin{lem}{hexedgecount}
Let $P\!\HEXA$ be a simple hexagonal
grid polygon, and let $P_1, P_2$ and $Q$ be defined as 
earlier. The number of edges satisfy the equation
$$E(P_1) + E(P_2) = E(P\!\HEXA) + E(Q)\;.$$
\end{lem}
\begin{proof*}
Obviously, two arbitrary polygons $P_1$ and $P_2$ always satisfy%
\footnote{Note that $P_1\cap P_2$ may have thin parts if common edges
of $P_1$ and $P_2$ do not occur in $P_1\cup P_2$. We count these
these edges twice in $P_1\cap P_2$.}
$$E(P_1) + E(P_2) = E(P_1\cup P_2) + E(P_1\cap P_2)\;.$$

With $P_1 \cap P_2 = P\!\HEXA\cap Q$ and 
$P_1\cup P_2=P\!\HEXA\cup Q$ we have
\begin{eqnarray*}
E(P_1) + E(P_2) & = & E(P_1\cap P_2) + E(P_1 \cup P_2)\\
& = & E(P\!\HEXA\cap Q) + E(P\!\HEXA\cup Q)\\
& = & E(P\!\HEXA) + E(Q)\qquad\qquad\proofendbox
\end{eqnarray*}
\end{proof*}

Finally, we need an upper bound for the length of a path inside a grid polygon.

\begin{lem}{hexshortest}
Let $\pfad$ be the shortest path between two cells 
in a hexagonal grid polygon $P\!\HEXA$.
The length of $\pfad$ is bounded by
$$|\pfad| \leq \frac14 E(P\!\HEXA) - \frac32\;.$$
\end{lem}
\begin{proof}
\Wlog\ we can assume that the start cell, $s$, and the target cell, $t$, of 
$\pfad$ belong to the first layer of $P$,
because we are searching for 
an upper bound for the shortest path between two arbitrary cells. 

Let $\pfad_1$ be the closed path in the 
1-offset of $P\!\HEXA$. For every forward step we have 2 edges
on the boundary of $P\!\HEXA$,
for every $60^\circ$ right turn 3 edges, 
for every $120^\circ$ right turn 4 edges, 
and for every left turn 1 edge.
For every left turn we can charge one right turn, so we have
in average 2 edges for every step. 
Further, we have 6 more right turns ($120^\circ$ turns count twice)
than left turns; that is,
$|\pfad_1| \leq {E(P\!\HEXA)-6 \over 2}$.
Now, observe the path $\pfad_L$ from $s$ to $t$ in the first layer that follows
the boundary of $P$ clockwise and the path $\pfad_R$ that follows
the boundary counterclockwise. 
In the worst case, both $\pfad_L$ and $\pfad_R$ have the same length, so
$|\pfad| = |\pfad_L| = |\pfad_R|$ holds. Thus, we have
$$2\cdot|\pfad| \leq |\pfad_1| \leq {E(P\!\HEXA)-6 \over 2} 
\;\Longrightarrow\;  
|\pfad| \leq \frac14 E(P\!\HEXA)-\frac32\,.$$
\end{proof}

\noindent
Now, we are able to show our main theorem:
\begin{theo}{hexmain}
Let $P\!\HEXA$ be a simple grid polygon with 
$C(P\!\HEXA)$ cells and $E(P\!\HEXA)$ edges. 
$P\!\HEXA$ can be explored with
$$S(P\!\HEXA)\leq C(P\!\HEXA)+\frac14 E(P\!\HEXA)-\frac52$$
steps. This bound is tight.
\end{theo}
\begin{proof}
We show by an induction on the number of components 
that $\excess(P\!\HEXA) \leq \frac14 E(P\!\HEXA)-\frac52$ holds.
For the induction base we consider a polygon without any split cell:
SmartDFS\HEXA\ visits each cell and returns on the shortest path 
to the start cell. Because there is no polygon split, all cells of $P\!\HEXA$ 
can be visited by a path of length $C(P\!\HEXA)-1$. 
By \reflem{hexshortest} the shortest path back to the start cell
is not longer than $\frac14 E(P\!\HEXA) - \frac32$; thus,
$\excess(P\!\HEXA)\leq \frac14 E(P\!\HEXA)-\frac52$ holds.

Now, we assume that there is more than one component during the 
application of SmartDFS\HEXA. 
Let $c$ be the first split cell detected in $P\!\HEXA$. 
When SmartDFS\HEXA\ reaches $c$, two new components, $K_1$ and $K_2$, occur.  
We consider the two polygons $P_1$ and $P_2$ defined as earlier,
using the polygon $Q$ around $c$;  
$K_2$ is recursively explored first.
As shown in \reflem{hexcomponent} we have
$$\excess(P\!\HEXA)\leq \excess(P_1)+ \excess(K_2\cup\{c\})+1\; .$$
Now, we apply the induction hypothesis to $P_1$ and $K_2\cup\{c\}$
and get
$$\excess(P\!\HEXA)
\leq \frac14 E(P_1)-\frac52+\frac14 E(K_2\cup\{c\})-\frac52 +1\,.$$
Applying \reflem{hexoffset} to the $q$-offset $K_2\cup\{c\}$ of $P_2$ 
yields $E(K_2\cup\{c\}) \leq E(P_2)-12q$. Thus,
we achieve 
$$\excess(P\!\HEXA) 
\leq \frac14 E(P_1)+\frac14 E(P_2) -3q -4\,.$$
With \reflem{hexedgecount} we have $E(P_1)+E(P_2)= E(P\!\HEXA) + E(Q)$;
and from \reflem{hexoffset} we conclude $E(Q) = 12q+6$ 
(a hexagon has 6 edges and $Q$ gains 12 edges per layer).
Altogether, we have 
\begin{eqnarray*}
\excess(P\!\HEXA) 
& \leq & \frac14 \Big(E(P_1)+ E(P_2)\Big) -3q -4\\
& \leq & \frac14 \Big(E(P\!\HEXA)+ 12q+6\Big) -3q -4\\
&=&\frac14 E(P\!\HEXA) - \frac52\,.
\end{eqnarray*}

It is easy to see that this bound is exactly achieved in corridors of width~1.
The exploration of such a corridor needs $2(C(P\!\HEXA)-1)$ steps.
On the other hand, the number of edges is $E(P\!\HEXA) = 4C(P\!\HEXA)+2$ 
(a corridor with one cell has 6 edges, and for every additional cell we
get 4 additional edges).
\end{proof}  

\sssect{Competitive Factor}{cfhex}~\\
So far we have shown an upper bound on the number of steps needed
to explore a polygon that depends on the number of cells and
edges in the polygon. Now, we want to analyze SmartDFS\HEXA\
in the competitive framework.

Corridors of width 1 or 2 play a crucial role in the following,
so we refer to them as {\em narrow passages}. More precisely, 
a cell, $c$, belongs to a narrow passage, if $c$ can be removed without
changing the layer number of any other cell. 

It is easy to see that narrow passages are explored optimally:
In corridors of width 1 both SmartDFS\HEXA\ and the optimal strategy
visit every cell twice, and in the other case both strategies visit
every cell exactly once. 

We need two lemmata to show a competitive factor for SmartDFS\HEXA.
The first one gives us a relation between the number of cells 
and the number of edges for a special class of polygons.

\begin{lem}{hexNoNarrow}
For a simple grid polygon, $P\!\HEXA$, with $C(P\!\HEXA)$ cells 
and $E(P\!\HEXA)$ edges, and without
any narrow passage or split cells in the first layer, we have
$$E(P\!\HEXA) \leq \frac43\,C(P\!\HEXA)+\frac{26}3\,.$$
\end{lem}
\begin{proof}
Consider a simple polygon, $P\!\HEXA$. We successively remove
at least three connected boundary cells that either form a straight
line or two lines with a $60^\circ$ angle. We remove a set of cells only
if the resulting polygon still fulfills our assumption that
the polygon has no narrow passages or split cells in the first layer.
These assumptions ensure that we can always find such a row or column
(i.e., if we cannot find such a row or column, the polygon has a narrow
passage or a split cell in the first layer).
Thus, we remove at {\em least} three cells and at most {\em four} edges. This decomposition
ends with a 'honeycomb', a center cell with its 6 neighbors, 
with $7$ cells and 18 edges
that fulfills $E(P\!\HEXA) = \frac43\,C(P\!\HEXA)+\frac{26}3$;
see \reffig{figs/hexCompLayer}(i)
Now, we reverse our decomposition; that is, we successively add all rows and
columns until we end up with $P$. In every step, we add at least three
cells and at most four edges. Thus, 
$E(P\!\HEXA) \leq \frac43\,C(P\!\HEXA)+\frac{26}3$
is fulfilled in every step.
\end{proof}

\pstexfig{(i) The minimal polygon that has neither narrow passages nor split
cells in the first layer, 
(ii) for polygons without narrow passages or split cells in the first 
layer, the last explored cell, $c'$, lies in the 1-offset, $P'$ (shaded).}
{figs/hexCompLayer}

For the same class of polygons, we can show that SmartDFS\HEXA\ behaves
slightly better than the bound in \reftheo{hexmain}.

\begin{lem}{hexStepNoNarrow}
A simple hexagonal grid polygon, $P\!\HEXA$, with $C(P\!\HEXA)$ cells 
and $E(P\!\HEXA)$ edges, and without
any narrow passage or split cells in the first layer can be 
explored using no more steps than
$$S(P\!\HEXA) \leq C(P\!\HEXA) + \frac14 E(P\!\HEXA) - \frac92\,.$$
\end{lem}
\begin{proof}
In \reftheo{hexmain} we have seen that 
$S(P\!\HEXA) \leq C(P\!\HEXA) +\frac14E(P\!\HEXA) - \frac52$
holds. To show this theorem, we used 
\reflem{hexshortest} on \refpagelem{hexshortest}
as an upper bound for
the shortest path back from the last explored cell to the start cell.
\reflem{hexshortest} bounds the shortest path from a cell, $c$, in the
first layer of $P\!\HEXA$ to the cell $c'$ that maximizes the distance to 
$c$ inside $P\!\HEXA$; thus, $c'$ is located in the first layer of 
$P\!\HEXA$, too.

Because $P\!\HEXA$ 
has neither narrow passages nor split cells in the first layer,
we can explore the first layer of $P\!\HEXA$ completely before we visit another
layer, see \reffig{figs/hexCompLayer}(ii). 
Therefore, the last explored cell, $c'$, of $P\!\HEXA$ is located in the 
1-offset of $P\!\HEXA$. Let $P'$ denote the 1-offset of $P\!\HEXA$, 
and $s'$ the first
visited cell in $P'$. Remark that $s$ and $s'$ are neighbors,
so the shortest path from $s'$ to $s$ is only one step.
Now, the shortest path, $\pfad$, from $c'$ to $s$ in $P\!\HEXA$ is bounded by 
a shortest path, $\pfad'$, from $c'$ to $s'$ in $P'$ and a
shortest path from $s'$ to $s$:
$$|\pfad| \leq |\pfad'| + 1\,.$$
The path $\pfad'$, in turn, is bounded using \reflem{hexshortest} by
$$|\pfad'| \leq \frac14 E(P')- \frac32\,.$$
By \reflem{hexoffset} (\refpagelem{hexoffset}), $E(P') \leq E(P\!\HEXA) -12$ 
holds, and altogether we get
$$|\pfad| \leq \frac14 E(P\!\HEXA) - \frac72\,,$$
which is two steps shorter than stated in \reflem{hexshortest}.
\end{proof}

\medskip

\noindent
Now, we can prove the following
\begin{theo}{hexcomp}
The strategy {SmartDFS\HEXA} is $\frac43$-competitive.
\end{theo}
\begin{proof}
Let $P\!\HEXA$ be a simple grid polygon. In the first stage, we remove
all narrow passages from $P\!\HEXA$ and get a sequence of (sub-)polygons
$P_i$, $i=1,\ldots,k$, without narrow passages.
For every $P_i$, $i=1,\ldots,k-1$, the optimal strategy 
in $P\!\HEXA$ explores the part of $P\!\HEXA$ 
that corresponds to $P_i$ up to the
narrow passage that connects $P_i$ with $P_{i+1}$, enters $P_{i+1}$, and
fully explores every $P_j$ with $j\geq i$. Then it returns to $P_i$ and
continues with the exploration of $P_i$. Further, we already know that
narrow passages are explored optimally. This allows us to 
consider every $P_i$ separately without changing the competitive factor 
of $P\!\HEXA$.

Now, we observe a (sub-)polygon $P_i$. We show by induction on the
number of split cells in the first layer that 
$S(P_i) \leq \frac43C(P_i) -\frac73 $ holds. 
Note that this bound is exactly achieved in a polygon as shown
in \reffig{figs/hexComp}: For the middle part---a corridor of width 3---,
SmartDFS\HEXA\ needs four steps for three cells. Additionally, we have
seven cells (the first row and the last two rows) that are explored
optimally. Thus, we have 
$\frac43\Big(C(P_i)-7\Big)+7=\frac43C(P_i)-\frac73$ steps.

\pstexfig{The competitive factor of $\frac43$ is exactly achieved:
In a corridor of width 3,
$S(P\!\HEXA)=\frac43\, \rmsub{S}{Opt}(P\!\HEXA)-\frac73$ holds.}
{figs/hexComp}

If $P_i$ has no split cell in the first layer (induction base), 
we can apply 
\reflem{hexStepNoNarrow} and \reflem{hexNoNarrow}:
\begin{eqnarray*}
S(P_i) &\leq & C(P_i)+\frac14\,E(P_i)-\frac92 \\
& \leq & C(P_i)+\frac14\left(\frac43\, C(P_i)+\frac{26}3\right)-\frac92\\
& = & \frac43\, C(P_i) - \frac73\,.
\end{eqnarray*}

\pstexfig{Two cases of split cells, (i) component of type II, (ii) 
component of type I.}
{figs/hexSplitComp}

Two cases occur if we meet a split cell, $c$, in the first layer, see
\reffig{figs/hexSplitComp}. In the first case, the new component
was never visited before (component of type II, see \refpage{comptypes}).
Here, we define $Q:=\{c\}$.
The second case occurs, because the explorer meets a cell, $c'$, 
that is in the first layer and touches the current cell, $c$, 
see for example \reffig{figs/hexSplitComp}(ii).
In this case, we define $Q$ as $\{c,c'\}$.

Similar to the proof of \reftheo{hexmain}, we split the polygon
$P_i$ into two parts, both including $Q$. Let $P''$ denote 
the part that includes the component of type I or II, $P'$ the
other part. For $|Q|=1$, see \reffig{figs/hexSplitComp}(i),
we conclude $S(P_i) = S(P')+S(P'')$ and
$C(P_i) = C(P')+C(P'')-1$. Applying the induction hypothesis to
$P'$ and $P''$ yields
\begin{eqnarray*}
S(P_i) & = & S(P')+S(P'')\\
&\leq & \frac43\,C(P')-\frac73 + \frac43\,C(P'')-\frac73\\
&=&  \frac43\,C(P_i)+\frac43-\frac{14}3 \quad 
< \quad \frac43\,C(P_i)-\frac73,.
\end{eqnarray*}

For $|Q|=2$ we gain some steps by merging the polygons.
If we consider $P'$ and $P''$ separately, we count the steps from $c'$ to
$c$---or vice versa---in both polygons, but in $P_i$ the path from $c'$
to $c$ is replaced by the exploration path in $P''$. Thus, we
have $S(P_i) = S(P')+S(P'')-2$ and 
$C(P_i) = C(P')+C(P'')-2$. This yields
\begin{eqnarray*}
S(P_i) & = & S(P')+S(P'')-2\\
&\leq & \frac43\,C(P')-\frac73 + \frac43\,C(P'')-\frac73-2\\
&=&  \frac43\,C(P_i)+\frac83-\frac73-2 \quad 
< \quad \frac43\,C(P_i)-\frac73\,.
\end{eqnarray*}

The optimal strategy needs at least $C$ steps, which, altogether, yields a 
competitive factor of $\frac43$. 
This factor factor is achieved in a 
corridor of width 3, see \reffig{figs/hexComp}.
\end{proof}

%% file: SmartDFSAnalysisTria.tex
\ssect{The Analysis of SmartDFS\TRIA}{trianalysis}
In this section, we analyze the performance of our strategy on
a triangular grid. 
The proof follows the same outline as the proof in the preceding 
section. The basic idea is an induction on the split cells. Thus,
we point out only the differences between SmartDFS\HEXA\ and
SmartDFS\TRIA.
The first difference concerns the $\ell$-offset:

\begin{lem}{trioffset}
The $\ell$-offset of a simple, triangular grid 
polygon, $P\!\TRIA$, has at least $6\ell$ edges fewer than $P\!\TRIA$.
\end{lem} 

\pstexfig{The 2-offset (shaded) of a grid polygon $P\!\TRIA$.}
{figs/tridoffset}

\begin{proof}
As in \reflem{hexoffset},
we can cut off blind alleys that are narrower than $2\ell$, because
those parts of $P\!\TRIA$ do not affect the $\ell$-offset.
We walk clockwise around the boundary of the remaining polygon, see
\reffig{figs/tridoffset}. For every $60^\circ$ left turn the offset gains 
at most $\ell$ edges and for every $60^\circ$ right turn the offset 
looses at least $\ell$ edges. 
Similar to the proof of \reflem{hexoffset}, there are
six more $60^\circ$ 
right turns than left turns (again, we count $120^\circ$ turns as 
two $60^\circ$ turns).
Altogether, we loose at least six edges for every layer.
\end{proof}

\pstexfig{A decomposition of $P\!\TRIA$: 
$K_1$ is of type III, $K_2$ of type II.}
{figs/triDecomposition1}

\pstexfig{A decomposition of $P\!\TRIA$: 
$K_1$ is of type III, $K_2$ of type I.
($Q$ may exceed $P\!\TRIA$.)}
{figs/triDecomposition2}

In line with SmartDFS\HEXA, we subdivide a polygon, $P\!\TRIA$, 
into three parts when we meet a split cell, $c$, in layer $\ell$:
$$P\!\TRIA = K_1 \disjoint K_2 \disjoint\ \{\,\mbox{visited cells of } P\!\TRIA\,\},$$
where $K_1$ and $K_2$ denote the connected components of the set of unvisited
cells, and $K_2$ is explored first.
Again, we have three possible types of components,
see \reffig{figs/triDecomposition1}
and \reffig{figs/triDecomposition2}:

\begin{romanlist}
\item[I.] $K_i$ is {\em completely} surrounded by layer $\ell$
\labelpage{tricomptypes}
\item[II.] $K_i$ is {\em not} surrounded by layer $\ell$
\item[III.] $K_i$ is {\em partially} surrounded by layer $\ell$
\end{romanlist}

\sssect{An Upper Bound on the Number of Steps}{ubtria}
As in the \refsect{ubtria}, 
$Q$ is the split cell broadened by $q$ layers, with
$$q:=\cases{\ell, & if $K_2$ is of type I\cr \ell-1, 
                  & if $K_2$ is of type II}\,.$$

\clearpage 

Further,
we choose $P_2 \subset P\!\TRIA\cup Q$ such that $K_2 \cup \{c\}$ 
is the $q$-offset of $P_2$, 
and $P_1 := ((P\!\TRIA\backslash P_2) \cup Q) \cap P\!\TRIA$.

\bigskip

%
%
We split the overall length of the exploration path, $\pfad$,
into two parts, $C(P\!\TRIA)$ and $\excess(P\!\TRIA)$, 
with $|\pfad|=C(P\!\TRIA)+\excess(P\!\TRIA)$.
\reflem{hexcomponent} and \reflem{hexedgecount}
hold also for triangular polygons:

\begin{lem}{tricomponent}
Let $P\!\TRIA$ be a simple triangular grid polygon. 
Let the agent visit the first 
split cell, $c$, which splits the unvisited cells of $P\!\TRIA$ 
into two components 
$K_1$ and $K_2$, where $K_2$ is of type I or II. 
With the preceding notations we have
$$\excess(P\!\TRIA)\leq \excess(P_1)+ \excess(K_2\cup\{c\})+1\; .$$
\end{lem}

\begin{lem}{triedgecount}
Let $P\!\TRIA$ be a simple triangular
grid polygon, and let $P_1, P_2$ and $Q$ be defined as 
earlier. The number of edges satisfy the equation
$$E(P_1) + E(P_2) = E(P\!\TRIA) + E(Q)\;.$$
\end{lem}

In contrast to hexagonal and square polygons, we need all but three edges
for an upper bound on the length of a shortest path:

\begin{lem}{trishortest}
Let $\pfad$ be the shortest path between two cells 
in a triangular grid polygon $P\!\TRIA$.
The length of $\pfad$ is bounded by
$$|\pfad| \leq  E(P\!\TRIA) -3\;.$$
\end{lem}

\pstexfig{(i) Possible projections of an edge of $P'$ (bold, dashed)
onto the edges of $P\!\TRIA$,
(ii) the cell $c$ cannot project its northwestern edge onto the boundary
of $P\!\TRIA$, because $c_1$ and $c_2$ are already charging the
boundary edges of $P\!\TRIA$.}
{figs/TriProjections12}

\begin{proof}
Let $P'$ be the grid polygon that is defined by the cells that $\Pi$
visits. We prove:  $E(P') \leq E(P)$; that is, the number of edges
around the path cannot exceed the number of edges in $P\!\TRIA$.
Thus, we have to find for every edge in $P'$ a corresponding edge
in $P\!\TRIA$: We project every edge, $e$, of a cell, $c$, in $P'$
along the axes of the grid onto the boundary of $P\!\TRIA$---provided that
$e$ is not already an edge in $P\!\TRIA$.
For every edge in $P'$ there are two possible axes for the projection,
see \reffig{figs/TriProjections12}(i).

Now, let us assume, we would doubly charge one edge in $P\!\TRIA$.
Then there is an edge, $e$, of a cell, $c$, of $\Pi$ where both
axes are blocked by other cells, $c_1$ and $c_2$ on $\Pi$, that both
project edges onto the boundary of $P\!\TRIA$,
see \reffig{figs/TriProjections12}(ii).

Now, we have two cases: Either, $c$ lies in $\Pi$ between $c_1$ and $c_2$,
or outside the path segment between $c_1$ and $c_2$ (in this case,
let $c_2$ lie between $c_1$ and~$c$).

\pstexfig{(i) $c$ lies in $\Pi$ between $c_1$ and $c_2$, 
(ii) $c_2$ lies in $\Pi$ between $c$ and $c_1$.
In both cases the path can be shortened (dashed).}
{figs/TriProjections3}

In the first case, we can shorten the path between $c_1$ and $c_2$ by
moving straight from $c_1$ to the cell $c'$ that is adjacent to $c$ 
and $e$ ($c'$ must be a free cell in $P\!\TRIA$; otherwise, $e$ would
be an edge in $P\!\TRIA$); see \reffig{figs/TriProjections3}(i).
In the other case, we can shorten the path, too: We replace the path
between $c_1$ and $c$ over $c_2$ by the straight path from $c_1$ to $c$;
see \reffig{figs/TriProjections3}(ii).
(The straight path between $c_1$ and $c$ must be in $P\!\TRIA$; otherwise, 
there would be a blocked cell and an edge; thus, we would be able to project 
$e$ onto the boundary of $P\!\TRIA$). Altogether, we can shorten $\Pi$.
This is a contradiction to the assumption that $\Pi$ is a shortest path.
Thus, for every edge, $e'$ in $P\!\TRIA$ there is at most one edge in $P'$
that is projected onto $e'$ (note that we project only from the interior
of $P$ to the boundary); that is, $E(P') \leq E(P)$ holds.

\medskip
In a corridor $K$ of width 1 we have 
$C(K)=E(K)-2$ (this equation holds for a single, triangular cell, and for
every further cell in a corridor of with 1 we add one cell and 
precisely one edge). $P'$ is such a corridor. Thus, we have
$$|\Pi| = C(P') - 1 = E(P') - 3 \leq E(P) - 3\,.$$
\end{proof}

\bigskip\noindent
Now, we are able to show our main theorem:
\begin{theo}{trimain}
Let $P\!\TRIA$ be a simple triangular grid polygon with 
$C(P\!\TRIA)$ cells and $E(P\!\TRIA)$ edges. 
$P\!\TRIA$ can be explored with
$$S(P\!\TRIA) \leq C(P\!\TRIA)+ E(P\!\TRIA)-4$$
steps. This bound is tight.
\end{theo}
\begin{proof}
The outline of this proof is the same as in \reftheo{hexmain}.
We show by an induction on the number of components 
that $\excess(P\!\TRIA) \leq E(P\!\TRIA)-4$ holds.

For the induction base we consider a polygon without any split cell which
can be explored in $C(P\!\TRIA)-1$ steps.
By \reflem{trishortest} the shortest path back to the start cell
is bounded by $E(P\!\TRIA)-3$; thus,
$\excess(P\!\TRIA)\leq E(P\!\TRIA)-4$ holds.

Now, let $c$ be the first split cell detected in $P\!\TRIA$. 
When reaching $c$, we have the components $K_1$ and $K_2$; we
explore $K_2$ first.
$P_1$, $P_2$, and $Q$ are defined as earlier.
As shown in \reflem{tricomponent} we have
$$\excess(P\!\TRIA)\leq \excess(P_1)+ \excess(K_2\cup\{c\})+1\; .$$
Now, we apply the induction hypothesis to $P_1$ and $K_2\cup\{c\}$
and get
$$\excess(P\!\TRIA)\leq E(P_1)-4+ E(K_2\cup\{c\})-4 +1\; .$$
Applying \reflem{trioffset} to the $q$-offset $K_2\cup\{c\}$ of $P_2$ 
yields 
$$\excess(P\!\TRIA) \leq E(P_1)+E(P_2)-6q-7$$

From \reflem{trioffset} we conclude $E(Q) = 6q+3$ (A triangle has 3 edges 
and $Q$ gains 6 edges per layer). With \reflem{triedgecount}
we have $\excess(P\!\TRIA)\leq E(P)-4$.

\noindent
This bound is exactly achieved in a corridor of width 1.
\end{proof}

\sssect{Competitive Factor}{cftria}~\\
Again, we also want to analyze our strategy
in the competitive framework. As in the preceding section,
{\em narrow passages} (i.e., corridors of width 1 or 2)
are explored optimally, so we consider only polygons without
such narrow passages or split cells in the first layer.
Analogously to \reflem{hexNoNarrow} (see \reffig{figs/triCompLayer}(i)) and
\reflem{hexStepNoNarrow} we can show:

\begin{lem}{triNoNarrow}
For a simple grid polygon, $P\!\TRIA$, with $C(P\!\TRIA)$ cells 
and $E(P\!\TRIA)$ edges, and without
any narrow passage or split cells in the first layer, we have
$$E(P\!\TRIA) \leq \frac13\,C(P\!\TRIA)+\frac{14}3\,.$$
\end{lem}

\pstexfig{(i) Minimal polygons that have neither narrow passages nor split
cells in the first layer, 
(ii) for polygons without narrow passages or split cells in the first 
layer, the last explored cell, $c'$, lies in the 1-offset, $P'$ (shaded).}
{figs/triCompLayer}

\begin{lem}{triStepNoNarrow}
A simple triangular grid polygon, $P\!\TRIA$, with $C(P\!\TRIA)$ cells 
and $E(P\!\TRIA)$ edges, and without
any narrow passage or split cells in the first layer can be 
explored using no more steps than
$$S(P\!\TRIA) \leq C(P\!\TRIA) + E(P\!\TRIA) - 6\,.$$
\end{lem}
\begin{proof}
The proof is analogously to \reflem{hexStepNoNarrow}, but we 
have to count three step for the path from $s'$ to $s$; see
\reffig{figs/triCompLayer}(ii). Thus, we have
$|\pfad| \leq |\pfad'| + 3$.
With $|\pfad'| \leq E(P')- 3$ (by \reflem{trishortest}) and
$E(P') \leq E(P\!\TRIA) - 6$ (by \reflem{hexoffset}), 
we get $|\pfad| \leq E(P\!\TRIA) - 6$,
which is two steps shorter than the one used in the proof of
\reftheo{trimain}.
\end{proof}

\medskip

\noindent
Now, we can prove the following
\begin{theo}{tricomp}
The strategy {SmartDFS\TRIA} is $\frac43$-competitive.
\end{theo}
\begin{proof}
In line with \reftheo{hexcomp}, 
we remove all narrow passages from the given polygon, $P\!\TRIA$,
and get a sequence of (sub-)polygons
$P_i$, $i=1,\ldots,k$, without narrow passages.
Again, we consider such a (sub-)polygon $P_i$ and show by an induction on the
number of split cells in the first layer that
$S(P_i) \leq \frac43C(P_i) -\frac43$ holds. 


With no split cell in the first layer, 
we can apply 
\reflem{triNoNarrow} and \reflem{triStepNoNarrow}:
$S(P_i)
\leq C(P_i)+E(P_i)-6
\leq C(P_i)+\frac13 C(P_i)+\frac{14}3-6
= \frac43 C(P_i) - \frac43$.

\pstexfig{Three cases of split cells, (i) component of type II, (ii) 
and (iii) component of type I.}
{figs/triSplitComp}

If we meet a split cell, $c$, in the first layer, we have two cases.
Either, the new component was never visited before (see 
\reffig{figs/triSplitComp}(i)), or we touch a cell, 
$c'$, that was already visited,
see for example \reffig{figs/triSplitComp}(ii) and (iii).
In the first case, let $Q := \{c\}$; in the second case
let $Q$ enclose the shortest path from $c$ to $c'$.

Similar to \reftheo{hexcomp}, we split the polygon.
In each case, we have $C(P_i) = C(P')+C(P'')-|Q|$ and
$S(P_i) = S(P')+S(P'')-2(|Q|-1)$, because $Q$ is a corridor of with 1.
Applying the induction hypothesis to
$P'$ and $P''$ yields
\begin{eqnarray*}
S(P_i) & = & S(P')+S(P'')-2|Q|+2\\
&\leq & \frac43\,C(P')- \frac43 + \frac43\,C(P'')- \frac43-2|Q|+2\\
&=&  \frac43\,C(P_i)-\frac23|Q| - \frac23 \quad 
\leq \quad \frac43\,C(P_i)-\frac43\quad(|Q|\geq 1)\,.
\end{eqnarray*}

The optimal strategy needs at least $C$ steps, which, altogether, yields a 
competitive factor of $\frac43$. 
The factor is achieved in a polygon as shown in 
\reffig{figs/triComp}: for every two rows of 12 cells in the middle, 
SmartDFS\TRIA\ needs 16 steps. Additionally, we need 10 steps for the 
first and the last row of 5 cells each. Hence, for $2n+2$ rows
we have ${10+16n\over 10+12n}$ which converges to $\frac43$.
\end{proof}

\pstexfig{A polygon where the competitive factor of SmartDFS\TRIA\
is achieved exactely.}
{figs/triComp}

%% file: summary.tex
\ssect{Summary}{summary}

We considered the online exploration of hexagonal and
triangular grid polygons and adapted the strategy SmartDFS.

For hexagonal polygons we gave a lower bound of $\frac{14}{13}$ and 
showed that SmartDFS\HEXA\ explores
polygons with 
$C$ cells and $E$ edges using no more than $C+\frac14 E-2.5$ steps.
For triangular polygons we have a lower bound of $\frac76$ (matching the
lower bound for square polygons) and an upper bound of
$C+E-4$ on the number of steps. Further, we showed that both
strategies are $\frac43$-competitive, and that the analysis was tight
(i.e., there are polygons where a factor of $\frac43$ is
exactly achieved).

An interesting observation is that---although the all three problems
appear to be the same at first sight---there are some subtle differences,
that are caused by the differences in the 'connectivity' of the 
grids: there are no touching cells in hexagonal grids, which seems to make
the problem easier (and---in turn--- makes it harder to find a
lower bound). On the other hand, the lower number of neighboring cells 
in triangular grids allows the assumption that there are more steps
needed in these kind of polygons. Indeed, we 'need' all edges in 
triangular polygons for the upper bound on the number of steps, while
in square and hexagonal polygons a fraction of the edges is sufficient.
In this connection, one may ask whether it is appropriate to 
consider---basically---the same strategy for all types of grids, or
one should regard more grid-specific details in the design of an 
exploration strategy.

An interesting open problems is how to close the gap
between the upper bound and the lower bound on the competitive factors.